\documentclass[10pt,a4paper]{article}

\pdfoutput=1

\usepackage[utf8]{inputenc}

\usepackage[margin=25mm]{geometry}

\usepackage{amsfonts}
\usepackage{amsmath}
\usepackage{amsthm}
\usepackage{amssymb}
\usepackage{bbm}

\usepackage{graphicx}
\usepackage{tikz}
\usetikzlibrary{arrows}
\usepackage{subfigure}

\usepackage{natbib}
\bibpunct{(}{)}{;}{a}{,}{,}

\usepackage{enumitem}
\usepackage[boxed, vlined, linesnumbered]{algorithm2e}
\usepackage{authblk}

\newtheorem{theorem}{Theorem}
\newtheorem{lemma}{Lemma}

\newtheorem{corollary}{Corollary}
\newtheorem{proposition}{Proposition}

\theoremstyle{definition}
\newtheorem{definition}{Definition}

\theoremstyle{remark}
\newenvironment{subprop}{\begin{enumerate}[label=(\roman*), ref=(\roman*), noitemsep]}{\end{enumerate}}

\SetCommentSty{textit}
\SetAlFnt{\small}
\SetAlCapFnt{\footnotesize}

\newcommand{\spforall}{\; \forall \;}
\newcommand{\spexists}{\; \exists \;}
\newcommand{\spst}{\;|\;}

\newcommand{\mcI}{\ensuremath{\mathcal{I}}}

\newcommand{\supscr}[1]{\ensuremath{^{\mathrm{#1}}}}
\newcommand{\subscr}[1]{\ensuremath{_{\mathrm{#1}}}}

\DeclareMathOperator{\pa}{pa}
\DeclareMathOperator{\nb}{ne}

\DeclareMathOperator{\doop}{do}

\DeclareMathOperator*{\argmin}{arg\,min}
\DeclareMathAlphabet{\mathsc}{OT1}{cmr}{m}{sc}

\newcommand{\LexBFS}{\ensuremath{\mathsc{LexBFS}}}
\newcommand{\Rand}{\ensuremath{\mathsc{Rand}}}
\newcommand{\RandAdv}{\ensuremath{\mathsc{RandAdv}}}
\newcommand{\OptSingle}{\ensuremath{\mathsc{OptSingle}}}
\newcommand{\OptUnb}{\ensuremath{\mathsc{OptUnb}}}
\newcommand{\MaxNb}{\ensuremath{\mathsc{MaxNb}}}

\definecolor{graphcol1}{rgb}{0.89, 0.10, 0.11}
\definecolor{graphcol2}{rgb}{0.22, 0.49, 0.72}
\definecolor{graphcol3}{rgb}{0.30, 0.69, 0.29}
\definecolor{graphcol4}{rgb}{0.60, 0.31, 0.64}

\newlength{\edgelength}
\newcommand{\grarright}{\mathbin{\tikz[baseline] \draw[->] (0pt, 0.7ex) -- (\edgelength, 0.7ex);}}
\newcommand{\grarleft}{\mathbin{\tikz[baseline] \draw[<-] (0pt, 0.7ex) -- (\edgelength, 0.7ex);}}
\newcommand{\grline}{\mathbin{\tikz[baseline] \draw[-] (0pt, 0.7ex) -- (\edgelength, 0.7ex);}}

\newcommand{\threegraph}[6]{%
  \begin{tikzpicture}[baseline=(one.base)]
    \node[anchor=base east] (one) at (0, 0) {#1};
    \node[anchor=base west] (two) at (1.2, 0) {#3};
    \node[anchor=base] (three) at (0.6, -0.5) {#5};
    \ifthenelse{\equal{#2}{}}{}{\draw[#2] (one.mid east) -- (two.mid west);}
    \ifthenelse{\equal{#4}{}}{}{\draw[#4] (two) -- (three);}
    \ifthenelse{\equal{#6}{}}{}{\draw[#6] (three) -- (one);}
  \end{tikzpicture}
}

\newlength{\exgredge}
\newenvironment{fournodeex}{%
\begin{tikzpicture}[baseline=(v1.base)]
  \node (v1) at (0, \exgredge) {$1$};
  \node (v2) at (-\exgredge, 0) {$2$};
  \node (v3) at (\exgredge, 0) {$3$};
  \node (v4) at (0, -\exgredge) {$4$};
}
{\end{tikzpicture}}
\newenvironment{optsingleex}{%
  \begin{tikzpicture}[baseline=(v1.base)]
    \node[anchor=mid] (v1) at (-3\exgredge, 0) {$1$};
    \node[anchor=mid] (v2) at (-2\exgredge, 0) {$2$};
    \node[anchor=mid] (v3) at (-1\exgredge, 0) {$3$};
    \node[anchor=mid] (v4) at (-1\exgredge, \exgredge) {$4$};
    \node[anchor=mid] (v5) at (0, 0) {$5$};
    \node[anchor=mid] (v6) at (-108:\exgredge) {$6$};
    \node[anchor=mid] (v7) at (-36:\exgredge) {$7$};
    \node[anchor=mid] (v8) at (36:\exgredge) {$8$};
    \node[anchor=mid] (v9) at (108:\exgredge) {$9$};
}
{\end{tikzpicture}}

\begin{document}

\title{Two Optimal Strategies for Active Learning of Causal Models From Interventional Data}
\author{Alain Hauser}
\affil{Department of Biology, Bioinformatics, University of Bern, CH-3012 Bern, Switzerland}
\author{Peter Bühlmann}
\affil{Seminar for Statistics, ETH Zürich, CH-8092 Zürich, Switzerland}

\maketitle

\begin{abstract}
From observational data alone, a causal DAG is only identifiable up to Markov equivalence.  Interventional data generally improves identifiability; however, the gain of an intervention strongly depends on the intervention target, that is, the intervened variables.  We present active learning (that is, optimal experimental design) strategies calculating optimal interventions for two different learning goals.  The first one is a greedy approach using single-vertex interventions that maximizes the number of edges that can be oriented after each intervention.  The second one yields in polynomial time a minimum set of targets of arbitrary size that guarantees full identifiability.  This second approach proves a conjecture of \citet{Eberhardt2008Almost} indicating the number of unbounded intervention targets which is sufficient and in the worst case necessary for full identifiability.  In a simulation study, we compare our two active learning approaches to random interventions and an existing approach, and analyze the influence of estimation errors on the overall performance of active learning.
\end{abstract}

\noindent\textbf{Keywords:} Active learning, graphical model, causal inference, Bayesian network, interventions

\section{Introduction}
\label{sec:introduction}

Causal relationships between random variables are usually modeled by directed acyclic graphs (DAGs), where an arrow between two random variables, $X \grarright Y$, reveals the former ($X$) as a \emph{direct} cause of the latter ($Y$).  From observational data alone (that is \emph{passively} observed data from the undisturbed system), directed graphical models are only identifiable up to Markov equivalence, and arrow directions (which are crucial for the causal interpretation) are in general not identifiable.  Without the assumption of specific functional model classes and error distributions \citep{Peters2011Identifiability}, the only way to improve identifiability is to use interventional data for estimation, that is data produced under a perturbation of the system in which one or several random variables are forced to specific values, irrespective of the original causal parents.  Examples of interventions include random assignment of treatments in a clinical trial, or gene knockdown or knockout experiments in systems biology.

The investigation of observational Markov equivalence classes has a long tradition in the literature \citep{Verma1990Equivalence,Andersson1997Characterization,Spirtes2000Causation}.  \citet{Hauser2012Characterization} extended the notion of Markov equivalence to the interventional case and presented a graph-theoretic characterization of corresponding Markov equivalence classes for a given set of interventions (possibly affecting several variables simultaneously).  Recently, we presented strategies for actively learning causal models with respect to sequentially improving identifiability \citep{Hauser2012Two}.  One of the strategies greedily optimizes the number of orientable edges with single-vertex interventions, and one that minimizes the number of interventions at arbitrarily many vertices to attain full identifiability.  This paper is an extended version of our previous work: besides a more detailed presentation of the algorithms, we evaluate their performance in the absence and presence of estimation errors and compare them to competing methods, and finally provide proofs for the correctness of the algorithms.

Several approaches for actively learning causal models have been proposed during the last decade, Bayesian as well as non-Bayesian ones, optimizing different utility functions.  All these active learning strategies consider sequential improvement of identifiability, which is different from the more classical active learning setting where one aims for sequential optimization of estimation accuracy \citep{Settles2012Active}.  In the non-Bayesian setting, \citet{Eberhardt2008Almost} and \citet{He2008Active} considered the problem of finding interventions that guarantee full identifiability of all representatives in a given (observational) Markov equivalence class which is assumed to be correctly learned.  The approach of \citet{Eberhardt2008Almost} works with intervention targets of unbounded size.  We prove the conjecture of \citet{Eberhardt2008Almost} on the number of intervention experiments sufficient and in the worst case necessary for fully identifying a causal model, and provide an algorithm that finds such a set of interventions in polynomial time (\OptUnb, see Section \ref{sec:unbounded-target}).  \citet{He2008Active} restrict the considerations to single-vertex interventions.  They propose an iterative line of action for learning causal models: their method estimates the observational Markov equivalence class in a first step and then incorporates interventional data to decide about edge orientations in subsequent steps.  This is not favorable from a statistical point of view since interventional data also yields information about parts of the graph that are not adjacent to the intervened variable.  We will see in Section \ref{sec:evaluation} that we indeed get smaller estimation errors in the finite sample case if we do not decouple the estimation of the observational Markov equivalence class and that of edge directions.  Moreover, the approach of \citet{He2008Active} is not able to cope with a situation in which we have few or no observational data, in contrast to ours.  \citet{Meganck2006Learning} compare different utility functions for single-vertex interventions, but do not address algorithmic questions of efficiently calculating optima of the utility functions.  In the Bayesian setting, \citet{Tong2001Active} and \citet{Masegosa2013Interactive} uses entropy-based utility functions.  While the approach of \citet{Tong2001Active} only interacts with the system under investigation, the approach of \citet{Masegosa2013Interactive} use (error-free) expert knowledge.

This paper is organized as follows: in Section \ref{sec:model}, we specify our notation of causal models and formalize our learning goals.  In Section \ref{sec:background}, we summarize graph-theoretic background material upon which our active learning algorithms, presented in Section \ref{sec:optimal-targets}, are based.  In Section \ref{sec:evaluation}, we evaluate our algorithms in a simulation study.  The proofs of the theoretical results of Section \ref{sec:optimal-targets} can be found in the appendix.

\section{Model}
\label{sec:model}

We consider a causal model on $p$ random variables $(X_1, \ldots, X_p)$ described by a DAG $D$.  Formally, a causal model is a pair $(D, f)$, where $D$ is a DAG on the vertex set $V = [p] := \{1, \ldots, p\}$ which encodes the \textbf{Markov property} of the (observational) density $f$: $f(x) = \prod_{i = 1}^p f(x_i \spst x_{\pa_D(i)})$; $\pa_D(i)$ denotes the parent set of vertex $i$.

Our notation and definitions related to graphs are summarized in Section \ref{sec:background}.  Unless stated otherwise, all graphs in this paper are assumed to have the vertex set $[p]$.

\subsection{Causal Calculus}
\label{sec:causal-calculus}

Beside the conditional independence relations of the observational density implied by the Markov property, a causal model also makes statements about effects of \textbf{interventions}.  We consider \textbf{stochastic interventions} \citep{Korb2004Varieties} modeling the effect of setting or forcing one or several random variables $X_I := (X_i)_{i \in I}$, where $I \subset [p]$ is called the \textbf{intervention target}, to the value of \emph{independent} random variables $U_I$.  Extending the $\doop()$ operator \citep{Pearl1995Causal} to stochastic interventions, we denote the \textbf{interventional density} of $X$ under such an intervention by
$$
    f(x | \doop_D(X_I\!=\!U_I))\!:=\!\prod_{i \notin I}\!f(x_i | x_{\pa_D(i)}) \prod_{i \in I}\!\tilde{f}(x_i),
$$
where $\tilde{f}$ is the density of $U_I$ on $\mathcal{X}_I$.  We also encompass the observational case as an intervention target by using $I = \emptyset$ and the convention $f(x | \doop(X_\emptyset = U_\emptyset)) = f(x)$.  The interventional density $f(x | \doop_D(X_I = U_I))$ has the Markov property of the \textbf{intervention graph} $D^{(I)}$, the DAG that we get from $D$ by removing all arrows pointing to vertices in $I$.  An illustration is given in Figure \ref{fig:ex-intervention-dags}.

\begin{figure}
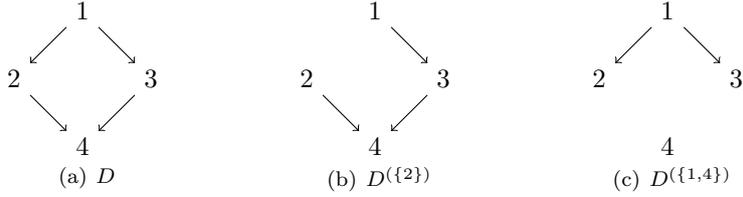

  \centering
  \subfigure[$D$]{%
    \begin{fournodeex}
      \draw[->] (v1) -- (v2);
      \draw[->] (v1) -- (v3);
      \draw[->] (v2) -- (v4);
      \draw[->] (v3) -- (v4);
    \end{fournodeex}
  } \qquad \qquad
  \subfigure[$D^{(\{2\})}$]{%
    \begin{fournodeex}
      \draw[->] (v1) -- (v3);
      \draw[->] (v2) -- (v4);
      \draw[->] (v3) -- (v4);
    \end{fournodeex}
  } \qquad \qquad
  \subfigure[$D^{(\{1, 4\})}$]{%
    \begin{fournodeex}
      \draw[->] (v1) -- (v2);
      \draw[->] (v1) -- (v3);
    \end{fournodeex}
  } \quad
  \caption{A DAG $D$ and the corresponding intervention graphs $D^{(\{2\})}$ and $D^{(\{1, 4\})}$.}
  \label{fig:ex-intervention-dags}
\end{figure}

We consider experiments based on data sets originating from \emph{multiple} interventions.  The \textbf{family of targets} $\mcI \subset \mathcal{P}([p])$, where $\mathcal{P}([p])$ denotes the power set of $[p]$, lists all (distinct) intervention targets used in an experiment.  A family of targets $\mcI = \{\emptyset, \{2\}, \{1, 4\}\}$ for example characterizes an experiment in which observational data as well as data originating from an intervention at $X_2$ and data originating from a (simultaneous) intervention at $X_1$ and $X_4$ are measured; that means an experiment in which data from all three DAGs in Figure \ref{fig:ex-intervention-dags} are collected.  In the whole paper, \mcI{} always stands for a family of targets with the property that for each vertex $a \in [p]$, there is some target $I \in \mcI$ in which $a$ is \emph{not} intervened ($a \notin I$).  This assumption, which is for example fulfilled when observation data is available ($\emptyset \in \mcI$), makes sure that the identifiability of the causal model under investigation is at least as good as in the observational case \citep{Hauser2012Characterization}.  Two DAGs $D_1$ and $D_2$ are called \textbf{\mcI-Markov equivalent} ($D_1 \sim_\mcI D_2$) if they are statistically indistinguishable under an experiment consisting of interventions at the targets in \mcI; we refer to \citet{Hauser2012Characterization} for a more formal treatment.
\begin{theorem}[\citet{Hauser2012Characterization}]
  \label{thm:markov-equivalence}
  Two DAGs $D_1$ and $D_2$ are \mcI-Markov equivalent if and only if
  \begin{subprop}
    \item $D_1$ and $D_2$ have the same skeleton and the same v-structures (that is, induced subgraphs of the form $a \grarright b \grarleft c$), and
    \item $D_1^{(I)}$ and $D_2^{(I)}$ have the same skeleton for all $I \in \mcI$.
  \end{subprop}
\end{theorem}
In the case of observational data, represented by $\mcI = \{\emptyset\}$, the second condition is a trivial consequence of the first one; in this case, Theorem \ref{thm:markov-equivalence} corresponds to the theorem of \citet{Verma1990Equivalence} characterizing observational Markov equivalence.

An \mcI-Markov equivalence class of a DAG $D$ is uniquely represented by its \textbf{\mcI-essential graph} $\mathcal{E_I}(D)$ \citep{Hauser2012Characterization}.  This partially directed graph has the same skeleton as $D$; a directed edge in $\mathcal{E_I}(D)$ represents \textbf{\mcI-essential} arrows, that is arrows that have the same orientation in all DAGs of the equivalence class; an undirected edge represents arrows that have different orientations in different DAGs of the equivalence class (see Figure \ref{fig:ex-essential-graph} for an example).  The concept of \mcI-essential graphs generalizes the one of CPDAGs which is well-known in the observational case \citep{Spirtes2000Causation,Andersson1997Characterization}.  We denote the \mcI-Markov equivalence class corresponding to an \mcI-essential graph $G$ by $\mathbf{D}(G)$.

\begin{figure}
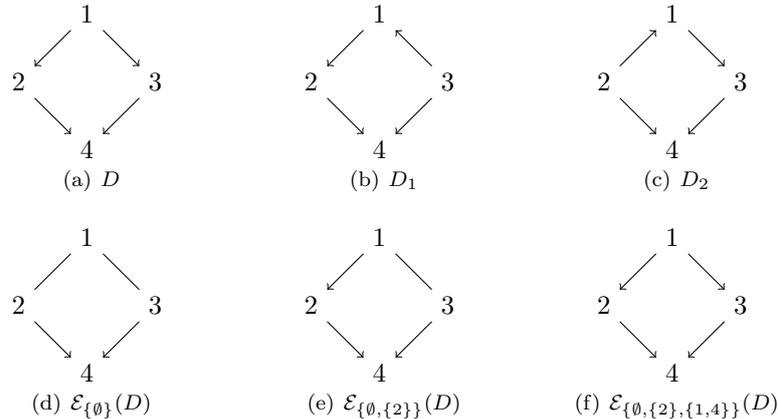

  \centering
  \subfigure[$D$]{%
    \qquad
    \begin{fournodeex}
      \draw[->] (v1) -- (v2);
      \draw[->] (v1) -- (v3);
      \draw[->] (v2) -- (v4);
      \draw[->] (v3) -- (v4);
    \end{fournodeex}
    \qquad
  } 
  \subfigure[$D_1$]{%
    \qquad
    \begin{fournodeex}
      \draw[->] (v1) -- (v2);
      \draw[->] (v2) -- (v4);
      \draw[->] (v3) -- (v1);
      \draw[->] (v3) -- (v4);
    \end{fournodeex}
    \qquad
  }
  \subfigure[$D_2$]{%
    \qquad
    \begin{fournodeex}
      \draw[->] (v1) -- (v3);
      \draw[->] (v2) -- (v1);
      \draw[->] (v2) -- (v4);
      \draw[->] (v3) -- (v4);
    \end{fournodeex}
    \qquad
  }
  
  \subfigure[$\mathcal{E}_{\{\emptyset\}}(D)$]{%
    \qquad
    \begin{fournodeex}
      \draw[-] (v1) -- (v2);
      \draw[-] (v1) -- (v3);
      \draw[->] (v2) -- (v4);
      \draw[->] (v3) -- (v4);
    \end{fournodeex}
    \qquad
  }
  \subfigure[$\mathcal{E}_{\{\emptyset, \{2\}\}}(D)$]{%
    \qquad
    \begin{fournodeex}
      \draw[->] (v1) -- (v2);
      \draw[-] (v1) -- (v3);
      \draw[->] (v2) -- (v4);
      \draw[->] (v3) -- (v4);
    \end{fournodeex}
    \qquad
  }
  \subfigure[$\mathcal{E}_{\{\emptyset, \{2\}, \{1, 4\}\}}(D)$]{%
    \qquad
    \begin{fournodeex}
      \draw[->] (v1) -- (v2);
      \draw[->] (v1) -- (v3);
      \draw[->] (v2) -- (v4);
      \draw[->] (v3) -- (v4);
    \end{fournodeex}
    \qquad
  }
  \caption{(a), (b), (c): the three DAGs of the observational Markov equivalence class of the DAG $D$ in Figure \ref{fig:ex-intervention-dags}.  (d), (e), (f): essential graphs of $D$ for different families of targets.  (d): with observational data alone ($\mcI = \{\emptyset\}$), the orientation of the upper two arrows are not identifiable.  (e): under the family of targets $\mcI = \{\emptyset, \{2\}\}$, $D_2$ is no longer \mcI-Markov equivalent to $D$, in contrast to $D_1$.  (f): under the family of targets $\mcI = \{\emptyset, \{2\}, \{1, 4\}\}$ finally, the DAG becomes completely identifiable.}
  \label{fig:ex-essential-graph}
\end{figure}

\subsection{Active Learning}
\label{sec:active-learning}

Assume $G$ is an \mcI-essential graph estimated from interventional data produced under the different interventions in \mcI.  We consider two different greedy active learning approaches.  In one step, the first one computes a single-vertex intervention that maximizes the number of orientable edges, while the second one computes an intervention target of arbitrary size that maximally reduces the clique number, that is the size of the largest clique of undirected edges (see Section \ref{sec:background}).  The motivation for the first approach is the attempt to quickly improve the identifiability of causal models with interventions at few variables; the motivation for the second approach is the conjecture of \citet{Eberhardt2008Almost} (which we prove in Corollary \ref{cor:intervention-number}) stating that maximally reducing the clique number after each intervention yields full identifiability of causal models with a minimal number of interventions.

Formally, our two algorithms yield a solution to the following problems.  The first one, called \OptSingle, computes a vertex
\begin{equation}
  v = \argmin_{v' \in [p]} \max_{D \in \mathbf{D}(G)} \xi\left(\mathcal{E}_{\mcI \cup \{\{v'\}\}}(D)\right) \ ,
  \label{eqn:objective-opt-single}
\end{equation}
where $\xi(H)$ denotes the number of undirected edges in a graph $H$.  The second algorithm, called \OptUnb, computes a set
\begin{equation}
  I = \argmin_{I' \subset [p]} \max_{D \in \mathbf{D}(G)} \omega\left(\mathcal{E}_{\mcI \cup \{I'\}}(D)\right) \ ,
  \label{eqn:objective-opt-unb}
\end{equation}
where $\omega(H)$ denotes the clique number of $H$ (see also Section \ref{sec:background}).  The key ingredients for the efficiency of \OptSingle{} (Algorithm \ref{alg:opt-single}) and \OptUnb{} (Algorithm \ref{alg:opt-unb}) are implementations that minimize the objective functions of Equation (\ref{eqn:objective-opt-single}) and (\ref{eqn:objective-opt-unb}), respectively, without enumerating all DAGs in the equivalence class represented by $G$.

We present the implementation of \OptSingle{} and \OptUnb{} in Section \ref{sec:optimal-targets}.  The next section provides the graph theoretic background upon which those implementations are based.

\section{Graph Theoretic Background}
\label{sec:background}

In this section, we give a brief summary of our notation (mostly following \citet{Andersson1997Characterization}) and basic facts concerning graphs.  We present classical results concerning graph coloring and perfect elimination orderings upon which our algorithms, especially \OptUnb{} (see Section \ref{sec:unbounded-target}), are based.

A \textbf{graph} is a pair $G = (V, E)$, where $V$ is a set of vertices and $E \subset (V \times V) \setminus \{(a, a) | a \in V\}$ is a set of edges.  We always assume $V = [p] := \{1, 2, \ldots, p\}$ and let the vertices of a graph represent the $p$ random variables $X_1, \ldots, X_p$.

An edge $(a, b) \in E$ with $(b, a) \in E$ is called \textbf{undirected} (or a line) and denoted by $a \grline b$, whereas an edge $(a, b) \in E$ with $(b, a) \notin E$ is called \textbf{directed} (or an arrow) and denoted by $a \grarright b$.  We use the short-hand notation $a \grarright b \in G$ to denote $(a, b) \in E$, but $(b, a) \notin E$, and $a \grline b \in G$ to denote $(a, b) \in E$ and $(b, a) \in E$.  $G$ is called directed (or undirected, respectively) if all its edges are directed (or undirected, respectively).  A \textbf{cycle} of length $k \geq 2$ is a sequence of $k$ distinct vertices of the form $(a_0, a_1, \ldots, a_k = a_0)$ such that $(a_{i-1}, a_i) \in E$ for $i \in \{1, \ldots, k\}$; the cycle is \textbf{directed} if at least one edge is directed.

For a subset $A \subset V$ of the vertices of $G$, the \textbf{induced subgraph} on $A$ is $G[A] := (A, E[A])$, where $E[A] := E \cap (A \times A)$.  A \textbf{v-structure} is an induced subgraph of $G$ of the form $a \grarright b \grarleft c$.  The \textbf{skeleton} of a graph $G$ is the undirected graph $G^u$ we get by replacing all directed edges in $G$ by an undirected one; formally, $G^u := (V, E^u)$, with $E^u := E \cup \{(a, b) \spst (b, a) \in E\}$.  The \textbf{parents} of a vertex $a \in V$ are the vertices $\pa_G(a) := \{b \in V \spst b \grarright a \in G\}$, its \textbf{neighbors} are the vertices $\nb_G(a) := \{b \in V \spst a \grline b \in G\}$; the \textbf{degree} of $a$ is defined as $\deg(a) := |\{b \in V \spst (a, b) \in E \vee (b, a) \in E\}|$, the maximum degree of $G$ is $\Delta(G) := \max_{a \in V} \deg(a)$.  The notation $|\cdot|$ here stands for set cardinality.

An undirected graph $G = (V, E)$ is \textbf{complete} if all pairs of vertices are neighbors.  A \textbf{clique} is a subset of vertices $C \subset V$ such that $G[C]$ is complete.  The \textbf{clique number} $\omega(G)$ of $G$ is the size of the largest clique in $G$.  $G$ is \textbf{chordal} if every cycle of length $k \geq 4$ contains a \textbf{chord}, that is two nonconsecutive adjacent vertices.

A \textbf{directed acyclic graph} or \textbf{DAG} is a directed graph without cycles.  A partially directed graph $G = (V, E)$ is a \textbf{chain graph} if it contains no \emph{directed} cycle; undirected graphs and DAGs are special cases of chain graphs.  Let $G'$ be the undirected graph we get by removing all directed edges from a chain graph $G$.  The \textbf{chain component} $T_G(a)$ of a vertex $a$ is the set of all vertices that are connected to $a$ in $G'$.  The set of all chain components of $G$ is denoted by $\mathbf{T}(G)$; they form a partition of $V$.  We extend the clique number to chain graphs $G$ by the definition $\omega(G) := \max_{T \in \mathbf{T}(G)} \omega(G[T])$.  To simplify notation, we will also use the term ``chain component'' to denote the subgraph $G[T]$ induced by an actual chain component $T$ when it is clear from the context whether the vertex set or the induced subgraph is meant.

An \textbf{ordering} of a graph $G$, that is a permutation $\sigma: [p] \to [p]$, induces a total order on $V$ by the definition $a \leq_\sigma b :\Leftrightarrow \sigma^{-1}(a) \leq \sigma^{-1}(b)$.  An ordering $\sigma = (v_1, \ldots, v_p)$ is a \textbf{perfect elimination ordering} if for all $i$, $\nb_{G^u} \cap \{v_1, \ldots, v_{i-1}\}$ is a clique in $G^u$.  A \textbf{topological ordering} of a DAG $D$ is an ordering $\sigma$ such that $a \leq_\sigma b$ for each arrow $a \grarright b \in D$; we then say that $D$ is \textbf{oriented according to $\sigma$}.  In our context, perfect elimination orderings play an important role for two reasons:
\begin{subprop}
  \item An \mcI-essential graph is a chain graph with chordal chain components \citep[Proposition 15]{Hauser2012Characterization}.  In order to construct a representative, its chain components have to be oriented according to a perfect elimination ordering \citep[Proposition 16]{Hauser2012Characterization}; this ensures that the representative has no v-structures which are not present in the \mcI-essential graph (see also Proposition \ref{prop:lex-bfs-peo}).
  
  \item As explained at the end of this section, greedy coloring along a perfect elimination ordering of a chordal graph yields an optimal coloring (Corollary \ref{cor:chordal-graph-perfect-coloring}).  This property is exploited by the algorithm \OptUnb{} (Algorithm \ref{alg:opt-unb}).
\end{subprop}
Algorithmically, we can generate a perfect elimination ordering on a chordal graph in linear time using a variant of the breadth-first search (BFS) called \textbf{lexicographic BFS} or \LexBFS{} \citep{Rose1970Triangulated}.  Let $G = (V, E)$ be an undirected graph for the rest of this section.  \LexBFS{} takes an ordering $(v_1, \ldots, v_p)$ of $V$ and the edge set $E$ as input and that outputs an ordering $\sigma = \LexBFS((v_1, \ldots, v_p), E)$ listing the vertices of $V$ in the visited order.  When the first $k$ vertices of the input ordering $(v_1, \ldots, v_p)$ form a clique in $G$, the output ordering $\sigma$ also starts with $v_1, \ldots, v_k$; we refer to \citet{Hauser2012Characterization} for details of such an implementation.  It is often sufficient to specify the input ordering of \LexBFS{} up to arbitrary orderings of some subsets of vertices.  For a set $A = \{a_1, \ldots, a_k\} \subset V$ and an additional vertex $v \in V \setminus A$, for example, we use the notation $\LexBFS((A, v, \ldots), E)$ to denote a \LexBFS-ordering produced from an input ordering of the form $(a_1, \ldots, a_k, v, \ldots)$, without specifying the orderings of $A$ and $V \setminus (A \cup \{v\})$.

\begin{proposition}[\citet{Rose1976Algorithmic}]
  \label{prop:lex-bfs-peo}
  Let $G = (V, E)$ be an undirected chordal graph with a \LexBFS-ordering $\sigma$.  Then $\sigma$ is also a perfect elimination ordering on $G$.  By orienting the edges of $G$ according to $\sigma$, we get a DAG without v-structures.
\end{proposition}

Algorithm \ref{alg:opt-unb} is strongly based on graph colorings.  A \textbf{$k$-coloring} of an undirected graph $G = (V, E)$ is a map $c: V \to [k]$; the coloring $c$ is \textbf{proper} if $c(u) \ne c(v)$ for every edge $u \grline v \in G$.  We say that $G$ is \textbf{$k$-colorable} if it admits a proper $k$-coloring; the \textbf{chromatic number} $\chi(G)$ of $G$ is the smallest integer $k$ such that $G$ is $k$-colorable.  By greedily coloring the vertices of the graph (see Algorithm \ref{alg:greedy-coloring}), one gets a proper $k$-coloring with $k \leq \Delta(G) + 1$ in polynomial time \citep{Chvatal1984Perfectly}.

\begin{algorithm}[t]
  \caption{$\mathsc{GreedyColoring}(G, \sigma)$.  Greedy algorithm that yields a proper coloring of $G$ along an ordering $\sigma$.}
  \label{alg:greedy-coloring}
  \SetKwInOut{Input}{Input}
  \SetKwInOut{Output}{Output}
  \Input{$G = ([p], E)$: undirected graph; $\sigma = (v_1, \ldots, v_p)$: ordering of $G$.}
  \Output{A proper coloring $c: [p] \to \mathbb{N}$}
  $c(v_1) \leftarrow 1$\;
  \For{$i = 2$ to $p$}{%
    $c(v_i) \leftarrow \min\{k \in \mathbb{N} \spst k \ne c(u) \spforall u \in \{v_1, \ldots, v_{i-1}\} \cap \nb(v_i)\}$\;
  }
  \Return{$c$}\;
\end{algorithm}

For any undirected graph $G$, the bounds $\omega(G) \leq \chi(G) \leq \Delta(G) + 1$ hold.  $G$ is \textbf{perfect} if $\omega(H) = \chi(H)$ holds for every induced subgraph $H$ of $G$.  An ordering $\sigma$ of $G$ is \textbf{perfect} if for any induced subgraph $H$ of $G$, greedy coloring along the ordering induced by $\sigma$ yields an optimal coloring of $H$ (that is, a $\chi(H)$-coloring).
\begin{proposition}[\citet{Chvatal1984Perfectly}]
  \label{prop:perfect-ordering}
  An ordering $\sigma$ of an undirected graph $G$ is perfect if and only if $G$ has no induced subgraph of the form $a \grline b \grline c \grline d$ with $a <_\sigma b$ and $d <_\sigma c$.
\end{proposition}
It can easily be seen that a perfect elimination ordering fulfills the requirement of Proposition \ref{prop:perfect-ordering}; hence we get, together with Proposition \ref{prop:lex-bfs-peo}, the following corollary.
\begin{corollary}
  \label{cor:chordal-graph-perfect-coloring}
  \begin{subprop}
    \item \label{itm:perfect-elimination-ordering-perfect} A perfect elimination ordering on some graph $G$ is perfect.
    \item \label{itm:chordal-graph-perfectly-orderable} Any chordal graph has a perfect ordering.
  \end{subprop}
\end{corollary}
\begin{proposition}[\citet{Chvatal1984Perfectly}]
  \label{prop:perfectly-orderable-graphs-perfect}
  A graph with a perfect ordering is perfect.
\end{proposition}

\section{Optimal Intervention Targets}
\label{sec:optimal-targets}

An \mcI-essential graph is a chain graph with chordal chain components.  Their edges are oriented according to a perfect elimination ordering in the DAGs of the corresponding equivalence class; edge orientations of different chain components do not influence (that means, additionally restrict) each other (see \citet{Hauser2012Characterization}, and Theorem \ref{thm:essential-graph-characterization}, \ref{sec:proofs}).  We can thus restrict our search for optimal intervention targets to single chain components. We can even treat each chain component as an observational essential graph, as the following lemma shows.  Its proof, as well as the proofs of all other statements of this section, are postponed to \ref{sec:proofs}.
\begin{lemma}
  \label{lem:intervention-local-effect}
  Consider an \mcI-essential graph $\mathcal{E_I}(D)$ of some DAG $D$, and let $T \in \mathbf{T}(\mathcal{E_I}(D))$ be one of its chain components.  Furthermore, let $I \subset [p]$, $I \notin \mcI$, be an (additional) intervention target.  Then we have
  $$
    \mathcal{E}_{\mcI \cup \{I\}}(D)[T] = \mathcal{E}_{\{\emptyset, I \cap T\}}(D[T]) \ .
  $$
\end{lemma}

\subsection{Single-vertex Interventions}
\label{sec:single-vertex-target}

We start with the treatment of the first active learning approach mentioned in Section \ref{sec:active-learning}.  By virtue of the following lemma, the maximum in Equation (\ref{eqn:objective-opt-single}) can be calculated without enumerating all representative DAGs in the equivalence class $\mathbf{D}(G)$.  This is a key insight for the development of a feasible algorithm solving Equation (\ref{eqn:objective-opt-single}), \OptSingle{} (Algorithm \ref{alg:opt-single}).

\begin{algorithm}[b!]
  \caption{$\OptSingle(G)$: yields a solution of Equation (\ref{eqn:objective-opt-single}).}
  \label{alg:opt-single}
  \SetKwInOut{Input}{Input}
  \SetKwInOut{Output}{Output}
  \Input{$G = ([p], E)$: \mcI-essential graph.}
  \Output{An optimal intervention vertex $v \in [p]$ in the sense of Equation \ref{eqn:objective-opt-single}, or $\emptyset$ if $G$ only has directed edges.}
  $\xi_0 \leftarrow$ number of unoriented edges in $G$\;
  $v\subscr{opt} \leftarrow 0$; $\xi\subscr{opt} \leftarrow p^2$\;
  \For{$v = 1$ to $p$ \label{ln:opt-single-for-start}}{%
    $\xi\subscr{max} \leftarrow -1$\;
    \ForEach{clique $C \subset \nb_G(v)$ \label{ln:opt-single-foreach-start}}{%
      $\sigma \leftarrow \LexBFS((C, v, \ldots)), E[T_G(v)])$\;
      $D \leftarrow$ DAG with skeleton $G[T_G(v)]$, topological ordering $\sigma$ \label{ln:opt-single-orient}\;
      $G' \leftarrow \mathcal{E}_{\{\emptyset, \{v\}\}}(D)$ \label{ln:opt-single-G}\;
      $\eta \leftarrow$ number of arrows in $G'$; $\xi \leftarrow \xi_0 - \eta$ \label{ln:opt-single-calc-xi}\;
      \lIf{$\xi > \xi\subscr{max}$}{$\xi\subscr{max} \leftarrow \xi$} \label{ln:opt-single-foreach-end}\;
    }
    \lIf{$-1 < \xi\subscr{max} < \xi\subscr{opt}$}{$(v\subscr{opt}, \xi\subscr{opt}) \leftarrow (v, \xi\subscr{max})$} \label{ln:opt-single-for-end} \;
  }
  \lIf{$v\subscr{opt} \neq 0$}{\Return{$v\subscr{opt}$}}
  \lElse{\Return{$\emptyset$}}\;
\end{algorithm}

\begin{lemma}
  \label{lem:local-characterization}
  Let $G$ be an \mcI-essential graph, and let $v \in [p]$.  Assume $D_1$ and $D_2 \in \mathbf{D}(G)$ such that $\{a \in \nb_G(v) \spst a \grarright v \in D_1\} = \{a \in \nb_G(v) \spst a \grarright v \in D_2\} = C$.  Then we have $D_1 \sim_{\mcI'} D_2$ under the family of targets $\mcI' = \mcI \cup \{\{v\}\}$.
\end{lemma}
The following proposition proves the correctness of Algorithm \ref{alg:opt-single}; an illustration of the algorithm is given in Figure \ref{fig:ex-opt-single}.
\begin{proposition}
  \label{prop:opt-single-correct}
  Let $G$ be an \mcI-essential graph.  If $G$ has no undirected edge, $\OptSingle(G) = \emptyset$.  Otherwise, $\OptSingle(G)$ is a single vertex solving Equation \ref{eqn:objective-opt-single}.
\end{proposition}

\begin{figure}[t]
  \centering
  \subfigure[$G$]{%
    \begin{optsingleex}
      \draw (v1) -- (v2) -- (v3) -- (v5) -- (v6);
      \draw (v3) -- (v4);
      \draw (v7) -- (v5) -- (v8);
      \draw (v5) -- (v9);
    \end{optsingleex}
  }
  \qquad\qquad
  \subfigure[$D$]{%
    \begin{optsingleex}
      \draw[->] (v2) -- (v1);
      \draw[->] (v3) -- (v2);
      \draw[->] (v3) -- (v4);
      \draw[->] (v3) -- (v5);
      \draw[->] (v5) -- (v6);
      \draw[->] (v5) -- (v7);
      \draw[->] (v5) -- (v8);
      \draw[->] (v5) -- (v9);
    \end{optsingleex}
  }
  
  \subfigure[Family of targets \mcI, \mcI-essential graph, and table of the utility function (worst-case number of unorientable edges $\xi$) as a function of a potential single-vertex intervention target $v$.]{%
    \begin{tabular}{ccc}
      \mcI & \mcI-essential graph & worst case \# of unoriented edges \\\hline
      $\{\emptyset\}$ &
      \begin{optsingleex}
        \draw (v1) -- (v2) -- (v3) -- (v5) -- (v6);
        \draw (v3) -- (v4);
        \draw (v7) -- (v5) -- (v8);
        \draw (v5) -- (v9);
        \draw[thick] (v5) circle (3mm);
      \end{optsingleex} &
      \begin{tabular}{c|*{9}{c}}
        $v$ & $1$ & $2$ & $3$ & $4$ & $5$ & $6$ & $7$ & $8$ & $9$ \\\hline
        $\xi$ & $7$ & $6$ & $4$ & $7$ & $3$ & $7$ & $7$ & $7$ & $7$
      \end{tabular}
      \\
      $\{\emptyset, \{5\}\}$ &
      \begin{optsingleex}
        \draw (v1) -- (v2) -- (v3) -- (v4);
        \draw[->] (v3) -- (v5);
        \draw[->] (v5) -- (v6);
        \draw[->] (v5) -- (v7);
        \draw[->] (v5) -- (v8);
        \draw[->] (v5) -- (v9);
        \draw[thick] (v2) circle (3mm);
      \end{optsingleex} &
      \begin{tabular}{c|*{9}{c}}
        $v$ & $1$ & $2$ & $3$ & $4$ & $5$ & $6$ & $7$ & $8$ & $9$ \\\hline
        $\xi$ & $3$ & $2$ & $2$ & $3$ & $4$ & $4$ & $4$ & $4$ & $4$
      \end{tabular}
      \\
      $\{\emptyset, \{5\}, \{2\}\}$ &
      \begin{optsingleex}
        \draw[->] (v2) -- (v1);
        \draw[->] (v3) -- (v2);
        \draw (v3) -- (v4);
        \draw[->] (v3) -- (v5);
        \draw[->] (v5) -- (v6);
        \draw[->] (v5) -- (v7);
        \draw[->] (v5) -- (v8);
        \draw[->] (v5) -- (v9);
        \draw[thick] (v3) circle (3mm);
      \end{optsingleex} &
      \begin{tabular}{c|*{9}{c}}
        $v$ & $1$ & $2$ & $3$ & $4$ & $5$ & $6$ & $7$ & $8$ & $9$ \\\hline
        $\xi$ & $1$ & $1$ & $0$ & $0$ & $1$ & $1$ & $1$ & $1$ & $1$
      \end{tabular}
      \\
      $\{\emptyset, \{5\}, \{2\}, \{3\}\}$ &
      \begin{optsingleex}
        \draw[->] (v2) -- (v1);
        \draw[->] (v3) -- (v2);
        \draw[->] (v3) -- (v4);
        \draw[->] (v3) -- (v5);
        \draw[->] (v5) -- (v6);
        \draw[->] (v5) -- (v7);
        \draw[->] (v5) -- (v8);
        \draw[->] (v5) -- (v9);
      \end{optsingleex} &
    \end{tabular}
  }
  \caption{(a): observational essential graph $G$; (b): a representative $D$ thereof. (c): steps of \OptSingle{} (Algorithm \ref{alg:opt-single}) sequentially proposing intervention targets, starting with the observational essential graph $G$, under the assumption that $D$ is the true underlying DAG.  Proposed intervention targets are circled in the picture.}
  \label{fig:ex-opt-single}
\end{figure}

\subsection{Interventions at Targets of Arbitrary Size}
\label{sec:unbounded-target}

We now proceed to the solution of Equation (\ref{eqn:objective-opt-unb}).  The following proposition, which was already conjectured by \citet{Eberhardt2008Almost}, shows that the minimum in Equation (\ref{eqn:objective-opt-unb}) only depends on the clique number of $G$:
$$
  \min_{I' \subset [p]} \max_{D \in \mathbf{D}(G)} \omega\left(\mathcal{E}_{\mcI \cup \{I'\}}(D)\right) = \lceil \omega(G) / 2 \rceil\ .
$$
In view of Lemma \ref{lem:local-characterization}, it is again sufficient to restrict the consideration to single chain components:
\begin{proposition}
  \label{prop:optimal-unbounded-target}
  Let $G$ be an undirected, connected, chordal graph on the vertex set $V = [p]$; such a graph is an observational essential graph.
  \begin{subprop}
    \item \label{itm:clique-number-upper-bound} There is an intervention target $I \subset V$ such that for every DAG $D \in \mathbf{D}(G)$, we have
    $$
      \omega(\mathcal{E}_{\{\emptyset, I\}}(D)) \leq \lceil \omega(G)/2 \rceil \ .
    $$
    
    \item \label{itm:clique-number-lower-bound} For every intervention target $I \subset [p]$ there is a DAG $D \in \mathbf{D}(G)$ such that
    $$
      \omega(\mathcal{E}_{\{\emptyset, I\}}(D)) \geq \lceil \omega(G)/2 \rceil \ .
    $$
  \end{subprop}
\end{proposition}

The constructive proof (see \ref{sec:proofs}) shows that a minimizer $I$ of Equation (\ref{eqn:objective-opt-unb}) can be generated by means of an optimal coloring which we can get by greedy coloring along a \LexBFS-ordering (see Proposition \ref{prop:lex-bfs-peo} and Corollary \ref{cor:chordal-graph-perfect-coloring}); this justifies Algorithm \ref{alg:opt-unb}.  An illustration of the algorithm is given in Figure \ref{fig:ex-opt-unb}.

\begin{algorithm}[t]
  \caption{$\OptUnb(G)$: yields a solution of Equation (\ref{eqn:objective-opt-unb}); time complexity is $O(p + |E|)$.  The proof of correctness follows from the proof of Proposition \ref{prop:optimal-unbounded-target} (\ref{sec:proofs}).}
  \label{alg:opt-unb}
  \SetKwInOut{Input}{Input}
  \SetKwInOut{Output}{Output}
  \Input{$G = ([p], E)$: essential graph.}
  \Output{An optimal intervention target $I \subset [p]$ in the sense of Equation \ref{eqn:objective-opt-unb}.}
  $I \leftarrow \emptyset$\;
  \ForEach{$T \in \mathbf{T}(G)$}{%
    $\sigma \leftarrow \LexBFS(T, E[T])$\;
    $c \leftarrow \mathsc{GreedyColoring}(G[T], \sigma)$ \label{ln:opt-unb-coloring}\;
    $\omega \leftarrow \max_{v \in [p]} c(v); \ h \leftarrow \lceil \omega / 2 \rceil$\;
    $I \leftarrow I \cup c^{-1}([h])$\;
  }
  \Return{$I$}\;
\end{algorithm}

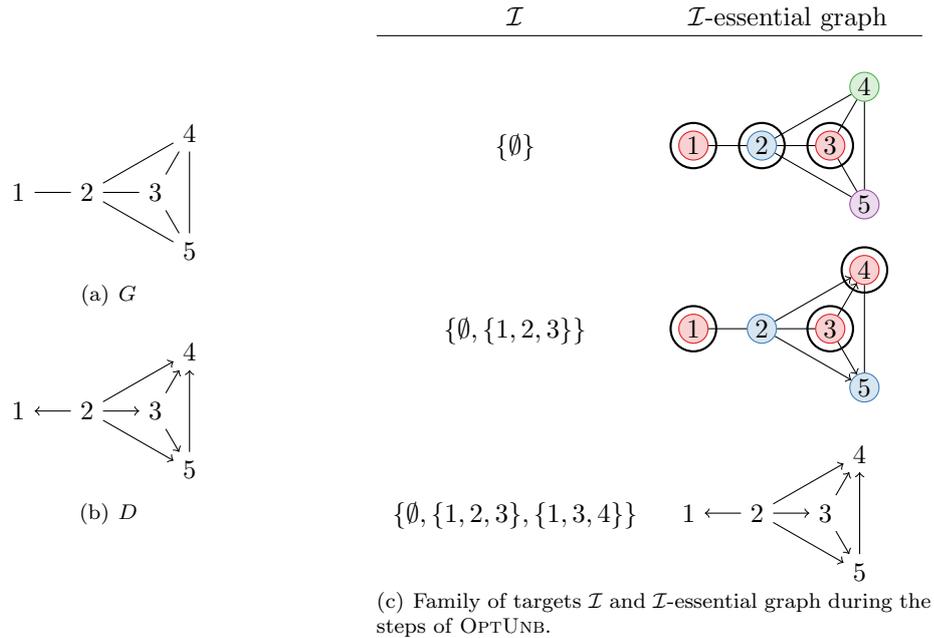
\begin{figure}
  \centering
  \begin{minipage}{0.3\linewidth}
    \subfigure[$G$]{%
      \begin{tikzpicture}
        \node (v1) at (-2\exgredge, 0) {$1$};
        \node (v2) at (-\exgredge, 0)  {$2$};
        \node (v3) at (0, 0)         {$3$};
        \node (v4) at (60:\exgredge)   {$4$};
        \node (v5) at (-60:\exgredge)  {$5$};
        
        \draw (v1) -- (v2) -- (v4) -- (v3) -- (v5) -- (v4);
        \draw (v3) -- (v2) -- (v5);
      \end{tikzpicture}
    }
    
    \subfigure[$D$]{%
      \begin{tikzpicture}
        \node (v1) at (-2\exgredge, 0) {$1$};
        \node (v2) at (-\exgredge, 0)  {$2$};
        \node (v3) at (0, 0)         {$3$};
        \node (v4) at (60:\exgredge)   {$4$};
        \node (v5) at (-60:\exgredge)  {$5$};
        
        \draw[->] (v2) -- (v1);
        \draw[->] (v2) -- (v3);
        \draw[->] (v2) -- (v4);
        \draw[->] (v2) -- (v5);
        \draw[->] (v3) -- (v4);
        \draw[->] (v3) -- (v5);
        \draw[->] (v5) -- (v4);
      \end{tikzpicture}
    }
  \end{minipage}
  \begin{minipage}{0.4\linewidth}
    \subfigure[Family of targets \mcI{} and \mcI-essential graph during the steps of \OptUnb{}.]{%
      \begin{tabular}{cc}
        \mcI & \mcI-essential graph \\\hline
        $\{\emptyset\}$ &
        \begin{tikzpicture}[baseline=(v1.base)]
          \tikzstyle{every node} = [circle, inner sep=1pt, minimum size=2mm]
          \node[fill=graphcol1!20, draw=graphcol1] (v1) at (-2\exgredge, 0) {$1$};
          \node[fill=graphcol2!20, draw=graphcol2] (v2) at (-\exgredge, 0)  {$2$};
          \node[fill=graphcol1!20, draw=graphcol1] (v3) at (0, 0)         {$3$};
          \node[fill=graphcol3!20, draw=graphcol3] (v4) at (60:\exgredge)   {$4$};
          \node[fill=graphcol4!20, draw=graphcol4] (v5) at (-60:\exgredge)  {$5$};
          \node (placeholder) at (\exgredge, 1.5\exgredge) {};
          
          \draw (v1) -- (v2) -- (v4) -- (v3) -- (v5) -- (v4);
          \draw (v3) -- (v2) -- (v5);
          
          \draw[thick] (v1) circle (3mm);
          \draw[thick] (v2) circle (3mm);
          \draw[thick] (v3) circle (3mm);
        \end{tikzpicture} \\
        $\{\emptyset, \{1, 2, 3\}\}$ &
        \begin{tikzpicture}[baseline=(v1.base)]
          \tikzstyle{every node} = [circle, inner sep=1pt, minimum size=2mm]
          \node[fill=graphcol1!20, draw=graphcol1] (v1) at (-2\exgredge, 0) {$1$};
          \node[fill=graphcol2!20, draw=graphcol2] (v2) at (-\exgredge, 0)  {$2$};
          \node[fill=graphcol1!20, draw=graphcol1] (v3) at (0, 0)         {$3$};
          \node[fill=graphcol1!20, draw=graphcol1] (v4) at (60:\exgredge)   {$4$};
          \node[fill=graphcol2!20, draw=graphcol2] (v5) at (-60:\exgredge)  {$5$};
          \node (placeholder) at (\exgredge, 1.5\exgredge) {};
          
          \draw (v1) -- (v2) -- (v3);
          \draw (v4) -- (v5);
          \draw[->] (v2) -- (v4);
          \draw[->] (v2) -- (v5);
          \draw[->] (v3) -- (v4);
          \draw[->] (v3) -- (v5);
          
          \draw[thick] (v1) circle (3mm);
          \draw[thick] (v3) circle (3mm);
          \draw[thick] (v4) circle (3mm);
        \end{tikzpicture} \\
        $\{\emptyset, \{1, 2, 3\}, \{1, 3, 4\}\}$ &
        \begin{tikzpicture}[baseline=(v1.base)]
          \node (v1) at (-2\exgredge, 0) {$1$};
          \node (v2) at (-\exgredge, 0)  {$2$};
          \node (v3) at (0, 0)         {$3$};
          \node (v4) at (60:\exgredge)   {$4$};
          \node (v5) at (-60:\exgredge)  {$5$};
          \node (placeholder) at (\exgredge, 1.5\exgredge) {};
          
          \draw[->] (v2) -- (v1);
          \draw[->] (v2) -- (v3);
          \draw[->] (v2) -- (v4);
          \draw[->] (v2) -- (v5);
          \draw[->] (v3) -- (v4);
          \draw[->] (v3) -- (v5);
          \draw[->] (v5) -- (v4);
        \end{tikzpicture}
      \end{tabular}
    }
  \end{minipage}

  \caption{(a): observational essential graph; (b): representative thereof.  (c): steps of \OptUnb{} (Algorithm \ref{alg:opt-unb}) sequentially proposing intervention targets, starting with the observational essential graph $G$, under the assumption that $D$ is the true underlying DAG.  Proposed intervention targets are circled in the picture; the coloring of the vertices in line \ref{ln:opt-unb-coloring} is shown.}
  \label{fig:ex-opt-unb}
\end{figure}

Since an \mcI-essential graph has only one representative DAG if and only if its clique number is $1$, a direct consequence of Proposition \ref{prop:optimal-unbounded-target} is a (sharp) upper bound on the number of interventions necessary to fully identify a causal model, as it was conjectured by \citet{Eberhardt2008Almost}.

\begin{corollary}
  \label{cor:intervention-number}
  Let $G$ be an \mcI-essential graph.  There is a set of $k = \lceil \log_2(\omega(G)) \rceil$ intervention targets $I_1, \ldots, I_k$ which are sufficient and in the worst case necessary to make the causal structure fully identifiable:
  $$
    \mathcal{E}_{\mcI \cup \{I_1, \ldots, I_k\}}(D) = D \spforall D \in \mathbf{D}(G).
  $$
\end{corollary}

The intervention targets $I_1, \ldots, I_k$ of Corollary \ref{cor:intervention-number} can be constructed by iteratively running Algorithm \ref{alg:opt-unb} on $G = \mathcal{E_I}(D)$, $\mathcal{E}_{\mcI \cup \{I_1\}}(D)$, $\mathcal{E}_{\mcI \cup \{I_1, I_2\}}(D)$ etc.  However, they could also be constructed at once by a modification of Algorithm \ref{alg:opt-unb}.  Let $c: [p] \to [\omega(G)]$ be a function such that for each chain component $T \in \mathbf{T}(G)$, $c|_T$ is a proper coloring of $G[T]$.  By defining $I_j$ as the set of all vertices whose color has a $1$ in the $j\supscr{th}$ position of its binary representation, we make sure that for every pair of neighboring vertices $u$ and $v$, there is at least one $j$ such that $|\{u, v\} \cap I_j| = 1$; hence the edge between $u$ and $v$ is orientable in $\mathcal{E}_{\mcI \cup \{I_1, \ldots, I_k\}}(D)$.  Since the binary representation of $\omega(G)$, the largest color in $c$, has length $k = \lceil \log_2(\omega(G)) \rceil$, this procedure creates a set of $k$ intervention targets that fulfill the requirements of Corollary \ref{cor:intervention-number}.

The problem of finding intervention targets to fully identify a causal model is related to the problem of finding separating systems of the chain components of essential graphs \citep{Eberhardt2007Causation}.  A \textbf{separating system} of an undirected graph $G = (V, E)$ is a subset $\mathcal{F}$ of the power set of $V$ such that for each edge $a \grline b \in G$, there is a set $F \in \mathcal{F}$ with $|F \cap \{a, b\}| = 1$.  \citet{Cai1984Separating} has shown that the minimum separating system of a graph $G$ has cardinality $\lceil \log_2(\chi(G)) \rceil$; together with the fact that clique number and chromatic number of the chordal chain components of essential graphs coincide (Corollary \ref{cor:chordal-graph-perfect-coloring} and Proposition \ref{prop:perfectly-orderable-graphs-perfect}), this also proves Corollary \ref{cor:intervention-number}.  The proof of \citet{Cai1984Separating} uses arguments similar to ours given in the paragraph above for the non-iterative determination of the targets $I_1, \ldots, I_k$ of Corollary \ref{cor:intervention-number}.

\subsection{Discussion}
\label{sec:algorithms-discussion}

\LexBFS{} and \textsc{GreedyColoring} have a time complexity of $O(p + |E|)$ when executed on a graph $G = ([p], E)$.  Thus, \OptUnb{} (Algorithm \ref{alg:opt-unb}) also has a linear complexity.\footnote{In contrast, finding a minimum separating set on \emph{non-chordal} graphs is NP-complete \citep{Cai1984Separating}.}  The time complexity of \OptSingle{} (Algorithm \ref{alg:opt-single}) on the other hand depends on the size of the largest clique in the \mcI-essential graph $G$.  By restricting \OptSingle{} to \mcI-essential graphs with a bounded vertex degree, its complexity is polynomial in $p$; otherwise, it is in the worst case exponential.

We emphasize that our two active learning algorithms do \emph{not} optimize the same objective; \OptUnb{} does \emph{not} guarantee maximal identifiability after each intervention, and \OptSingle{} does \emph{not} guarantee a minimal number of single-vertex interventions to full identifiability.

Consider for example the (observational) essential graph $G$ in Figure \ref{fig:ex-opt-single}(a).  All its representatives are fully identifiable after at most two single-vertex interventions: the first intervention should be performed at vertex $3$, the second one either at vertex $2$ or $5$.  This can be seen as follows: after intervening at vertex $3$, at least the edges to its $G$-neighbors can be oriented.  Since the true underlying DAG cannot have v-structures in this example (they would be identifiable even in the observational case, see Theorem \ref{thm:markov-equivalence}, and hence be present in the observational essential graph $G$), at most one of the arrows between vertex $3$ and its $G$-neighbors can point into $3$.  If we find $2 \grarright 3 \in D$, the line between $1$ and $2$ needs an additional intervention at vertex $2$ (or $1$) to be oriented; if we find $5 \grarright 3 \in D$, an additional intervention at vertex $5$ is sufficient to orient all edges between vertices $5$ to $9$.  If the DAG in Figure \ref{fig:ex-opt-single}(b) represents the true causal model, however, \OptSingle{} will need \emph{three} steps to full identifiability; it will iteratively propose interventions at targets $5$, $2$ and $3$ (see Figure \ref{fig:ex-opt-single}(c)).  Note that if interventions at several vertices yield the same number of unorientable edges, the algorithm chooses the candidate with the smallest index; this is the case in the second and third step in Figure \ref{fig:ex-opt-single}.

In general, \OptUnb{} does not yield an intervention target of minimal size.  With two straightforward improvements, we could reduce the number of intervened vertices: first, we could take $h \leftarrow \lfloor \omega/2 \rfloor$ instead of $h \leftarrow \lceil \omega/2 \rceil$ in Algorithm \ref{alg:opt-unb}; the proof of Proposition \ref{prop:optimal-unbounded-target} is also valid with this choice.  Second, we could permute the colors produced by the greedy coloring such that $|c^{-1}(\{1\})| \leq |c^{-1}(\{2\})| \leq \ldots$.  However, since an optimal coloring of a graph is not unique, not even up to permutation of colors, these heuristic improvements would still not guarantee an intervention target of minimal size with the properties required in Proposition \ref{prop:optimal-unbounded-target}.  In the example shown in Figure \ref{fig:ex-opt-unb}, the intervention targets chosen by \OptUnb{} do clearly not have minimal size: the alternative targets $\{4, 5\}$ and $\{2, 5\}$ (which would result from reordering the colors are proposed above) would have exactly the same effect on edge orientations as the chosen ones.

\section{Experimental evaluation}
\label{sec:evaluation}

We evaluated Algorithms \ref{alg:opt-single} and \ref{alg:opt-unb} in a simulation study on $4000$ randomly generated Gaussian causal models with vertex numbers $p \in \{10, 20, 30, 40\}$.  More details are provided in the next subsection.

\subsection{Methods and Models}
\label{sec:simulation-methods}

We considered two experimental settings which we refer to as ``oracle case'' and ``sample case''.

In the oracle case, we only considered the identifiability of causal models under a given family of intervention targets, neglecting estimation errors.  This setting, which corresponds to an infinite sample case, aims at solely evaluating the selection of intervention targets as performed by Algorithms \ref{alg:opt-single} and \ref{alg:opt-unb} in the absence of estimation errors.  We compared five algorithms proposing intervention targets: our algorithms \OptSingle{} and \OptUnb{}, a purely random proposition of single-vertex interventions (denoted by \Rand), a slightly advanced random approach that randomly chooses any vertex which has at least one incident undirected edge (denoted by \RandAdv), and a method introduced by \citet{He2008Active} choosing the vertex with the maximum number of neighbors (denoted by \MaxNb).  \citet{He2008Active} proposed this approach as a feasible heuristic to their exponential-time algorithm for finding a minimum set of single-vertex interventions leading to full identifiability of all DAGs in a Markov equivalence class.

In the sample case, we estimated interventional Markov equivalence classes from simulated data and applied the algorithms listed above on \emph{estimated} instead of \emph{true} interventional essential graphs.  This allows to investigate the influence of finite sample size and the corresponding estimation errors to the overall performance of active learning approaches.  In this setting, we evaluated \Rand{}, \RandAdv{}, \OptSingle{} and \OptUnb{} together with GIES (Greedy Interventional Equivalence Search) \citep{Hauser2012Characterization}, an interventional generalization of the GES algorithm of \citet{Chickering2002Optimal} that estimates an interventional essential graph from jointly observational and interventional data (possibly originating from different interventions).  We also considered the algorithm of \citet{He2008Active} that estimates the observational Markov equivalence class from observational data in a first step and orients previously unoriented edges by testing independence between an intervened variable and a neighbor in a second step.  To estimate the observational essential graph, we used the PC \citep{Spirtes2000Causation} and the GES \citep{Chickering2002Optimal} algorithm; since we simulated data from Gaussian distributions (see below), we used t-tests on a significance level of $5\%$ to detect dependences between an intervened and a neighoring vertex.

We always started from the observational essential graph (or an estimate thereof based on observational data) and proceeded by including the intervention targets proposed by the different active learning algorithms.  In the oracle case, all active learning approaches lead to full identifiability of the true causal model after a finite number of steps.  In the sample case, differences between estimated essential graph and true DAG are not only due to limited identifiability, but also due to estimation errors.  In particular, we can end up with a \emph{wrongly estimated DAG}: an estimated essential graph consisting of oriented edges only that does not correspond to the true DAG.  In this case, the active learning strategies \OptSingle, \OptUnb{}, \MaxNb{} and \RandAdv{} do not propose new targets any more.  To avoid getting stuck in wrong DAGs, we continued by drawing additional observational data points in order to get larger samples.

For each vertex number $p = \{10, 20, 30, 40\}$, we randomly generated $1000$ DAGs with a binomial distribution of vertex degrees, having an expected degree of $3$.  Data used in the sample case were drawn from Gaussian distributions.  A Gaussian causal model can be represented by an edge weight for each arrow and an error variance for each vertex of the DAG; we randomly drew edge weights from a uniform distribution on $[-1.5, -0.5] \cup [0.5, 1.5]$ and error variances from a uniform distribution on $[0.01, 0.2]$.  Interventions were performed with expectation $4$ times as large as the mean observational marginal standard deviations.  We used sample sizes of $n \in \{50, 100, 200, 500, 1000, 2000, 5000\}$ \emph{per intervention target}.  As soon as the algorithms claimed full identifiability (that is, ended up with an interventional Markov equivalence class with a single representative), we drew the same number of observational data points in the following step.  We used the procedure presented in \citet{Hauser2012Characterization} to draw observational or interventional samples from Gaussian causal models.

In order to assess the distance of estimated \mcI-essential graphs to the true DAG, we used the \textbf{structural Hamming distance} or \textbf{SHD} \citep{Brown2005Comparison} (in a slightly adapted version of \citet{Kalisch2007Estimating}), which is the sum of false positives of the skeleton, false negatives of the skeleton and wrongly oriented edges.  Furthermore, we defined the ``survival time'' of a DAG as the number of active learning steps needed for correct reconstruction of the true DAG, measured in intervention targets ($T$) or intervened variables ($V$).  If the structure of a  DAG was correctly learned for example under the family $\mcI = \{\emptyset, \{2\}, \{1, 4\}\}$, we counted $T = 2$ (non-empty) targets and $V = 3$ variables.  For each vertex number and algorithm, we estimated the ``survival function'', that is the probability $S_T(t) := P[T > t]$ or $S_V(v) := P[V > v]$, respectively, with a Kaplan-Meier estimator \citep{Kaplan1958Nonparametric}.

\subsection{Results}
\label{sec:simulation-results}

Figure \ref{fig:identifiability-survival} shows the survival functions of the active learning algorithms in the oracle case.  \Rand{} was clearly beaten by all other approaches, and \OptUnb{} clearly dominated all other strategies in terms of intervention targets.  However, if we measure the number of intervened vertices, \OptUnb{} was even slightly worse than \RandAdv, what is not surprising in view of the discussion of the size of intervention targets in Section \ref{sec:algorithms-discussion}.  \OptSingle{} gave a significant improvement over \RandAdv; however, the step from \Rand{} to \RandAdv{} is much larger than from \Rand{} to \OptSingle.  The performance of \MaxNb{} is not distinguishable from that of \OptSingle{} on a significance level of $5\%$ (Figure \ref{fig:identifiability-survival}, lower row).  This is a good reason to use the computationally simpler heuristic \MaxNb{} instead of \OptSingle{}.

\begin{figure}
  \centering
  \includegraphics{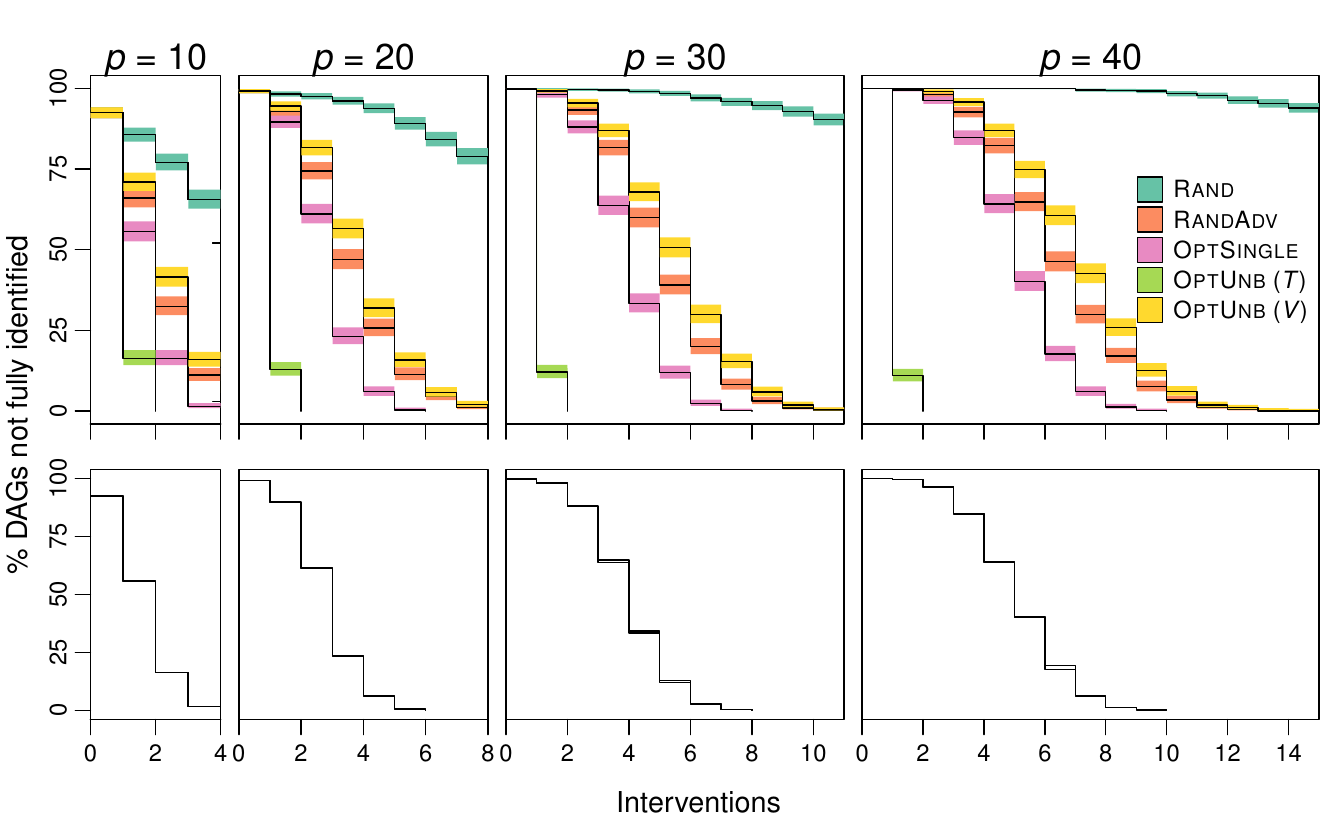}
  \caption{Number of intervention steps needed for full identifiability of DAGs in the oracle case, measured in targets (T) or intervened variables (V); for algorithms proposing only single-vertex interventions, both numbers are the same.  Thin lines: Kaplan-Meier estimates; shaded bands in the upper row: $95\%$ confidence region.  Lower row: survival curves of \OptSingle{} (lower line) and \MaxNb{} (upper line); they mostly coincide.}
  \label{fig:identifiability-survival}
\end{figure}
\begin{figure}
  \centering
  \includegraphics{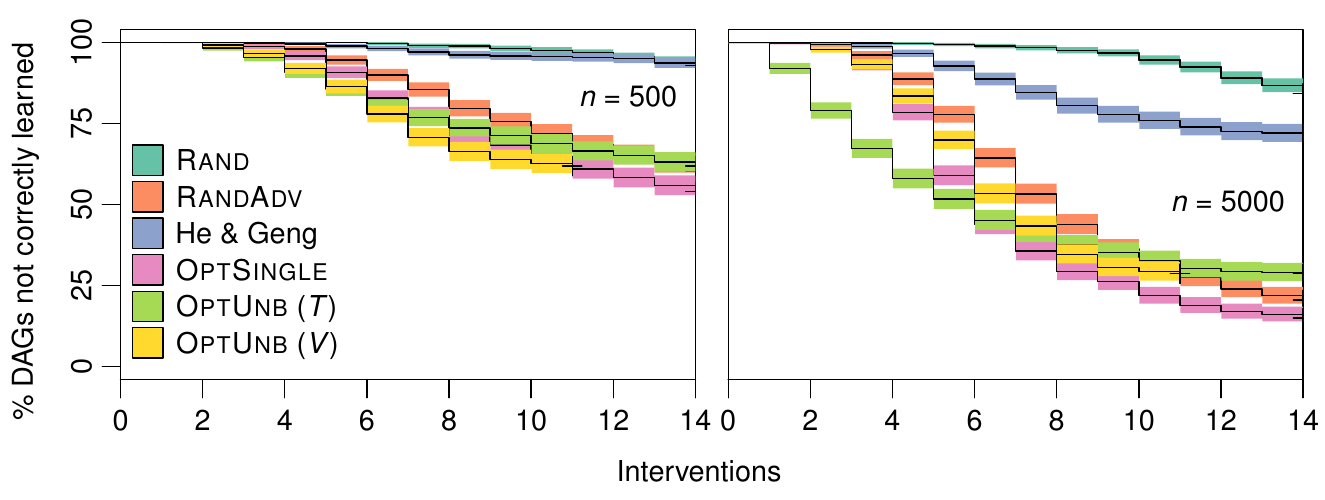}
  \caption{Number of intervention steps needed for correctly learning DAGs with $p = 30$ in the sample case, measured in targets (T) or intervened variables (V).  Thin lines: Kaplan-Meier estimates; shaded bands: $95\%$ confidence region.  Comparison of two different sample sizes per intervention: $n = 500$ and $n = 5000$.  GES was used for the observational estimation step in the algorithm of \citet{He2008Active}; the PC algorithm showed a worse performance (data not shown).  }
  \label{fig:survival-gies}
\end{figure}

As expected, the algorithms performed worse when applying active learning strategies to \emph{estimated} interventional essential graphs: the survival functions are flatter in the sample than in the oracle case (Figure \ref{fig:survival-gies}), and the performance differences between the algorithms shrink.  The active learning strategies proposed in this paper, namely combining \OptSingle{} or \OptUnb{} with GIES, show significantly better performance than the approach of \citet{He2008Active}; for $n = 500$ data points per intervention drawn from models with $p = 30$ for example, the approach of \citet{He2008Active} is almost not distinguishable from GIES used with randomly chosen intervention targets (left plot of Figure \ref{fig:survival-gies}).  The results reported in Figure \ref{fig:survival-gies} were found using GES for the observational estimation step; the PC algorithm showed a worse performance.  The significance level of the t-test used to orient edges turned out to have little relevance for the overall performance: runs with a significance level of $1\%$ and $10\%$ did not give significantly different results than the runs with a significance level of $5\%$ shown in the figure.  Since we have seen that the heuristic \MaxNb{} works as well as \OptSingle{} in the oracle case, the results in Figure \ref{fig:survival-gies} shows the advantage of having an estimation algorithm that considers the ensemble of interventional and observational data as a whole instead of decoupling the estimation into two stages, one only considering observational and one only interventional data.

The hierarchy of algorithms \OptSingle{}, \OptUnb{}, \Rand{} and \RandAdv{} found in the oracle case (Figure \ref{fig:identifiability-survival}) is still visible for GIES estimates with $n = 5000$ data points per intervention for, say, up to $7$ intervention steps.  The performances of \RandAdv{}, \OptSingle{} and \OptUnb{} however do not significantly differ any more with $n = 500$ data points per intervention.

The SHD between the \mcI-essential graph and the true DAG (Figure \ref{fig:shd-oracle-gies}) describes the difference between the estimated and the true models.  In the oracle case, the relative performance of the competing algorithms resembles that of Figure \ref{fig:identifiability-survival}.  In the sample case, however, the results strongly depend on sample size: for $n = 50$ samples per target, random interventions performed as well as optimally chosen ones in combination with GIES; for $n = 5000$ samples per target, the superiority of \OptSingle{} and \OptUnb{} over their random competitors became noticeable.  For all sample sizes, the approach of \citet{He2008Active} showed a markedly worse performance.

\begin{figure}
  \centering
  \includegraphics{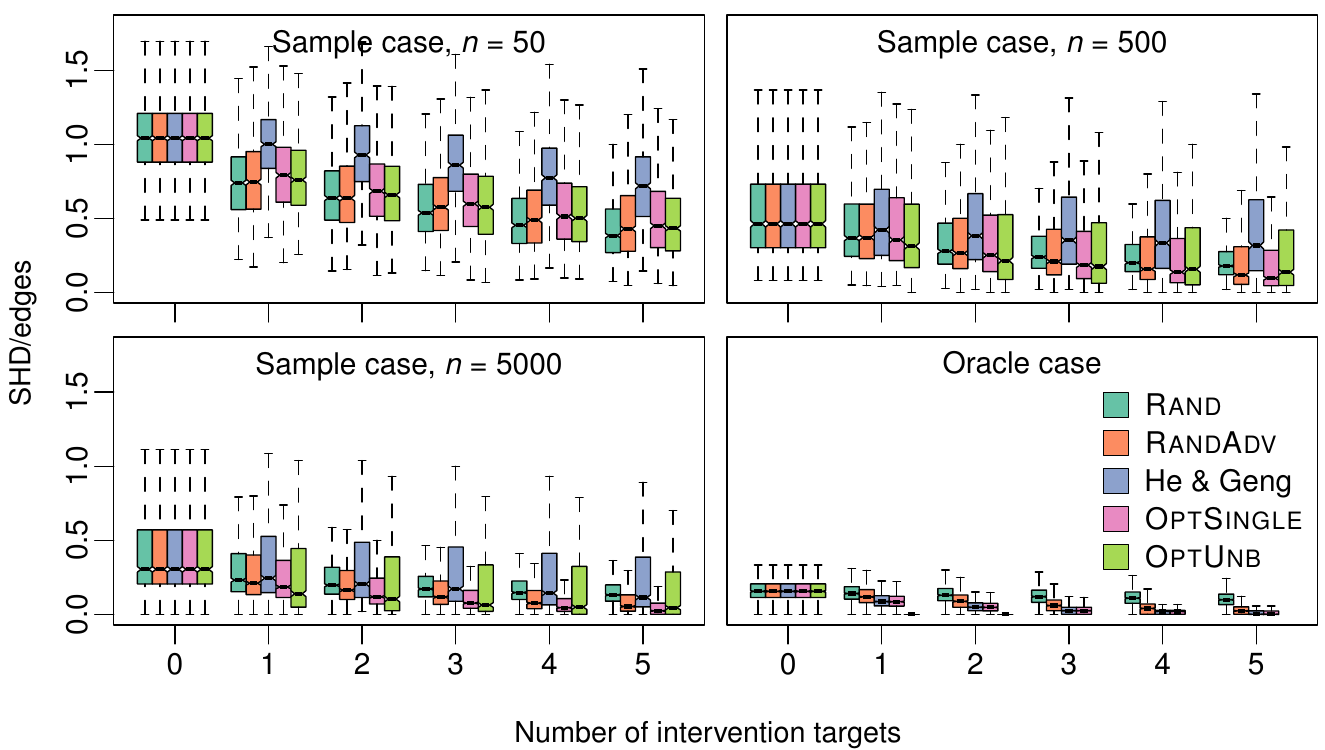}
  \caption{SHD to true DAG per edges of the true DAG as a function of the number of intervention targets.  Plots for DAGs with $p = 30$ vertices, using different active learning approaches and sample sizes.}
  \label{fig:shd-oracle-gies}
\end{figure}

\section{Conclusion}
\label{sec:conclusion}

We developed two algorithms which propose optimal intervention targets: one that finds the single-vertex intervention which maximally increases the number of orientable edges (called \OptSingle), and one that maximally reduces the clique number of the non-orientable edges with an intervention at arbitrarily many variables (called \OptUnb).  We proved a conjecture of \citet{Eberhardt2008Almost} concerning the number of interventions sufficient and in the worst case necessary for fully identifying a causal model by showing that \OptUnb{} yields, when applied iteratively, a \emph{minimum} set of intervention targets that guarantee full identifiability.

In a simulation study, we demonstrated that both algorithms lead significantly faster to full identifiability than randomly chosen interventions in the oracle case, that is, not taking estimation errors originating from finite sample size into account.  If we count the total number of intervened variables, however, \OptUnb{} performed slightly worse than a random approach.  This illustrates the fact that sequentially intervening single variables yields in general more identifiability that intervening those variables simultaneously.  With few samples in the range of $n \approx p$ per intervention, estimation errors dominated over limited identifiability.

In the finite sample case with reasonably large sample sizes, our approach clearly outperformed the algorithm of \citet{He2008Active}, demonstrating the benefit of an estimation algorithm operating on an ensemble of interventional and observational data.

\subsection*{Acknowledgments}

We thank Jonas Peters, Frederick Eberhardt and the anonymous reviewers for valuable comments on the manuscript.



\appendix

\section{Proofs}
\label{sec:proofs}

The proof of Lemma \ref{lem:intervention-local-effect} is rather technical and heavily based on the graph theoretic characterization of interventional essential graphs in Theorem \ref{thm:essential-graph-characterization} developed by \citet{Hauser2012Characterization}.  The theorem relies on the following notion of ``strongly protected arrows'':

\begin{definition}[Strong protection; \citet{Hauser2012Characterization}]
  \label{def:strongly-protected-arrow} Let $G$ be a (partially oriented) graph, and \mcI{} a family of intervention targets.  An arrow $a \grarright b \in G$ is \textbf{strongly \mcI-protected} in $G$ if there is some $I \in \mcI$ such that $|I \cap \{a, b\}| = 1$, or the arrow $a \grarright b$ occurs in at least one of the following four configurations as an induced subgraph of $G$:
  \begin{center} 
    (a): \threegraph{$a$}{->}{$b$}{}{$c$}{->} \
    (b): \threegraph{$a$}{->}{$b$}{<-}{$c$}{} \
    (c): \threegraph{$a$}{->}{$b$}{<-}{$c$}{<-} \
    (d):
    \begin{tikzpicture}[baseline=(a.base)]
      \node[anchor=base east] (a) at (0, 0) {$a$};
      \node[anchor=base west] (b) at (1.6, 0) {$b$};
      \node[anchor=base] (c1) at (0.8,  0.7) {$c_1$};
      \node[anchor=base] (c2) at (0.8, -0.7) {$c_2$};
      \draw[->] (a)  -- (b);
      \draw[->] (c1) -- (b);
      \draw[->] (c2) -- (b);
      \draw[-]  (a)  -- (c1);
      \draw[-]  (a)  -- (c2);
    \end{tikzpicture}
  \end{center}
\end{definition}

\begin{theorem}[\citet{Hauser2012Characterization}]
  \label{thm:essential-graph-characterization}
  A graph $G$ is the \mcI-essential graph of a DAG $D$ if and only if
  \begin{subprop}
    \item \label{itm:chain-graph} $G$ is a chain graph;
    
    \item \label{itm:chordal-chain-components} for each chain component $T \in \mathbf{T}(G)$, $G[T]$ is chordal;
    
    \item \label{itm:forbidden-subgraph} $G$ has no induced subgraph of the form $a \grarright b \grline c$;
    
    \item \label{itm:forbidden-edge} $G$ has no line $a \grline b$ for which there exists some $I \in \mcI$ such that $|I \cap \{a, b\}| = 1$;
    
    \item \label{itm:strongly-protected-arrows} every arrow $a \grarright b \in G$ is strongly \mcI-protected.
  \end{subprop}
\end{theorem}

\begin{proof}[Proof of Lemma \ref{lem:intervention-local-effect}]
  To shorten notation, we define $G := \mathcal{E}_{\mcI \cup \{I\}}(D)[T]$ and $H := \mathcal{E}_{\{\emptyset, I \cap T\}}(D[T])$.
  
  $G$ has the same skeleton as $D[T]$ and $H$.  Furthermore, since a DAG is contained in its essential graph in the graph theoretic sense, we have $D[T] \subset G$ and $D[T] \subset H$.  We conclude that the two graphs $G$ and $H$ cannot have arrows of opposite orientation.
  
  To prove the graph inclusion $G \subset H$, it remains to show that every undirected edge $a \grline b \in G$ is also undirected in $H$.
  
  Assume, for the sake of contradiction, that this is wrong.  Then there are vertices $a$ and $b$ with $a \grline b \in G$ and $a \grarright b \in H$.  We assume without loss of generality that $b$ is minimal in
  $$
    B := \{v \in [p] \spst \spexists u \in [p]: u \grline v \in G, u \grarright v \in H\}
  $$
  with respect to $\preceq_D$, the preorder on the vertices $[p]$ defined by $u \preceq_D v :\Leftrightarrow \exists$ path from $u$ to $v$ in $D$.
  
  Since $H = \mathcal{E}_{\{\emptyset, I \cap T\}}(D[T])$ is an interventional essential graph, the arrow $a \grarright b$ is strongly $\{\emptyset, I \cap T\}$-protected in $H$ (Theorem \ref{thm:essential-graph-characterization}\ref{itm:strongly-protected-arrows}, Definition \ref{def:strongly-protected-arrow}).  If $|I \cap \{a, b\}| = 1$, the edge between $a$ and $b$ would also be oriented in $G$, a contradiction.  Hence the arrow occurs in one of the four configurations of Definition \ref{def:strongly-protected-arrow} in $H$; we consider the four possibilities separately, showing that each of them leads to a contradiction; this proves our claim.
  \begin{enumerate}[label=(\alph*), nolistsep]
    \item The arrow $c \grarright a$ must also be present in $G$ by the minimality assumption on $b$.  This means that $c \grarright a \grline b$ is an induced subgraph of $G$, contradicting Theorem \ref{thm:essential-graph-characterization}\ref{itm:forbidden-subgraph}.
    
    \item By Theorem \ref{thm:markov-equivalence}(i), the v-structure $a \grarright b \grarleft c$ would also be present in $D[T]$ and hence in $G$, a contradiction.
    
    \item By the assumption of minimality of $b$, the arrow $a \grarright c$ must also be present in $G$.  This means that $\gamma = (a, c, b, a)$ is a \emph{directed} cycle in $G$ and hence in $\mathcal{E}_{\mcI \cup \{I\}}(D)$, contradicting Theorem \ref{thm:essential-graph-characterization}\ref{itm:chain-graph}.
    
    \item Consider the orientation of the arrows between $c_1$, $c_2$ and $a$ in $D$.  At least one of the arrows must point away from $a$, because there would be an additional v-structure otherwise.  Assume without loss of generality that $a \grarright c_1 \in D$; the argument is then the same as in configuration (c).
  \end{enumerate}

  We skip the proof the graph inclusion $G \supset H$ here; it can be done by very similar arguments.
\end{proof}

Lemma \ref{lem:local-characterization} is a simple consequence of Theorem \ref{thm:markov-equivalence}:

\begin{proof}[Proof of Lemmma \ref{lem:local-characterization}]
  From Theorem \ref{thm:markov-equivalence}, we know that $D_1$ and $D_2$ have the same skeleton and the same v-structures, and that the intervetnino DAGs $D_1^{(I)}$ and $D_2^{(I)}$ have the same skeleton for all $I \in \mcI$.  It remains to show that $D_1^{(\{v\})}$ and $D_2^{(\{v\})}$ have the same skeleton; this is indeed the case since all edges incident to $v$ have the same orientation in $D_1$ and $D_2$.
\end{proof}

Proposition \ref{prop:opt-single-correct} proves the correctness of Algorithm \ref{alg:opt-single}, \OptSingle.  Its proof is based on the fact that every clique in $\nb_G(v)$ is an admissible set $C$ in the sense of Lemma \ref{lem:local-characterization} and vice versa:
\begin{proposition}[\citet{Andersson1997Characterization}]
  \label{prop:lex-bfs-clique}
  Let $G$ be an undirected chordal graph, $a \in [p]$ and $C \subset \nb_G(a)$.  There is a DAG $D \subset G$ with $D^u = G$ and $\{b \in \nb_G(a) \spst b \grarright a \in D\} = C$ which is oriented according to a perfect elimination if and only if $C$ is a clique.
\end{proposition}

\begin{proof}[Proof of Proposition \ref{prop:opt-single-correct}]
  If $G$ has no undirected edges, the neighbor set of every vertex $v \in [p]$ is empty and the inner for loop in Algorithm \ref{alg:opt-single} (lines \ref{ln:opt-single-foreach-start} to \ref{ln:opt-single-foreach-end}) is never executed.  Hence $v\subscr{opt}$ is never assigned in the for-loop and \OptSingle{} returns the empty set.
  
  We assume that $G$ has at least one undirected edge from here on.  Then the inner for loop (lines \ref{ln:opt-single-foreach-start} to \ref{ln:opt-single-foreach-end}) of Algorithm \ref{alg:opt-single} is run at least once, and the algorithm will return a single vertex.
  
  In the rest of the proof, we confine our considerations to the inner for loop of Algorithm \ref{alg:opt-single} for a fix vertex $v$.  We must show that the value $\xi\subscr{opt}$ calculated in the loop fulfills
  $$
    \xi\subscr{opt} = \max_{D \in \mathbf{D}(G)} \xi\left(\mathcal{E}_{\mcI \cup \{\{v\}\}}(D)\right)\ .
  $$
  
  By the comments in the beginning of Section \ref{sec:optimal-targets}, an intervention at a vertex $v$ only improves the identifiability of edges in the chain component $T_G(v)$.  By Lemma \ref{lem:intervention-local-effect}, the chain component $T_G(v)$ can be treated as an observational essential graph; this is what Algorithm \ref{alg:opt-single} does in the inner for loop  (lines \ref{ln:opt-single-foreach-start} to \ref{ln:opt-single-foreach-end}).  The value $\xi$ calculated in line \ref{ln:opt-single-calc-xi} is therefore the number of unoriented edges in $\mathcal{E}_{\mcI \cup \{\{v\}\}}(D_1)$, where $D_1$ is a DAG we get by orienting the edges of $T_G(v)$ as in $D$ (line \ref{ln:opt-single-orient}) and the edges of the remaining chain components according to an arbitrary perfect elimination ordering.
  
  In summary, we calculate the maximum number of unorientable edges over a subset of $\mathbf{D}(G)$ in lines \ref{ln:opt-single-foreach-start} to \ref{ln:opt-single-foreach-end}, hence we have
  $$
    \xi\subscr{opt} \leq \max_{D \in \mathbf{D}(G)} \xi\left(\mathcal{E}_{\mcI \cup \{\{v\}\}}(D)\right)
  $$
  at the end of the inner for loop.
  
  We claim that the converse also holds:
  \begin{equation}
    \xi\subscr{opt} \geq \max_{D \in \mathbf{D}(G)} \xi\left(\mathcal{E}_{\mcI \cup \{\{v\}\}}(D)\right)\ .
    \label{eqn:maximum-bound}
  \end{equation}
  Assume that $D_2 \in \mathbf{D}(G)$ is a maximizer of the right hand side of (\ref{eqn:maximum-bound}).  Let $C := \{a \in \nb_G(v) \spst a \grarright v \in D_2\}$, and define $\sigma := \LexBFS((C, v, \ldots), E[T_G(v)])$.  Note that $C$ is a clique by Proposition \ref{prop:lex-bfs-clique}; hence the \LexBFS-ordering $\sigma$ is of the form $\sigma = (C, v, \ldots)$ (see Section \ref{sec:background}).  Furthermore, let $D_1$ be the DAG we get by orienting the edges of $G[T_G(v)]$ according to the topological ordering $\sigma$, and the edges of the remaining chain components according to an arbitrary perfect elimination ordering.  \ref{lem:local-characterization}; especially, their $\mcI'$-essential graphs coincide for the family of targets $\mcI' = \mcI \cup \{\{v\}\}$.  Hence we have
  $$
    \xi\left(\mathcal{E}_{\mcI \cup \{\{v\}\}}(D_1)\right) = \xi\left(\mathcal{E}_{\mcI \cup \{\{v\}\}}(D_2)\right)\ .
  $$
  Since $D_1$ is visited by Algorithm \ref{alg:opt-single}, this means that
  $$
    \xi\subscr{opt} \geq \max_{D \in \mathbf{D}(G)} \xi\left(\mathcal{E}_{\mcI \cup \{\{v\}\}}(D)\right)\ ,
  $$
  which proves the claim.
\end{proof}

We finally proceed to the proof of the main theoretic result of the paper.  Note that the proof of Proposition \ref{prop:optimal-unbounded-target} is constructive and also proves the correctness of Algorithm \ref{alg:opt-unb}.

\begin{proof}[Proof of Proposition \ref{prop:optimal-unbounded-target}]
  \begin{subprop}
    \item
      Since $G$ is chordal, we have $\chi(G) = \omega(G)$ by Corollary \ref{cor:chordal-graph-perfect-coloring}(ii) and Proposition \ref{prop:perfectly-orderable-graphs-perfect}.  Let $c: V \to [\omega(G)]$ be a proper coloring of $G$.  Define $I := c^{-1}([h])$ for $h := \lceil \omega(G)/2 \rceil$.  With an intervention at the target $I$, at most the edges of $G[I]$ and $G[V \setminus I]$ are unorientable for any causal structure $D \in \mathbf{D}(G)$ under the family of targets $\mcI := \{\emptyset, I\}$.  Therefore the bound
      $$
        \omega(\mathcal{E_I}(D)) \leq \max\{\omega(G[I]), \omega(G[V \setminus I])\}
      $$
      holds for every $D \in \mathbf{D}(G)$.  It remains to show that both of these terms are bounded by $h$.
      
      The induced subgraph $G[I]$ is also perfect, and $c|_I$ is a proper $h$-coloring of $G[I]$.  Hence we have $\omega(G[I]) = \chi(G[I]) \leq h$.  Analogously, $c|_{V \setminus I}$ is a proper $(\omega(G) - h)$-coloring of $G[V \setminus I]$, hence we have $\omega(G[V \setminus I]) = \chi(G[V \setminus I]) \leq \omega(G) - h \leq h$ by definition of $h$.
    
    \item
      Let $C$ be a maximum clique in $G$, and define $C \cap I =: \{v_1, \ldots, v_k\}$ and $C \setminus I =: \{v_{k+1}, \ldots, v_\omega\}$.  The \LexBFS-ordering
      $$
        \sigma := \LexBFS((v_1, \ldots, v_\omega, \ldots), E)
      $$
      starts with the vertices $v_1, \ldots, v_\omega$.  Set $\mcI := \{\emptyset, I\}$ and let $D \in \mathbf{D}(G)$ be oriented according to $\sigma$.
      
      We claim that the arrows in $D[C \cap I]$ and in $D[C \setminus I]$ are not \mcI-essential in $D$.  For $v_i, v_j \in C \cap I$ ($i < j$), consider the ordering
      $$
        \sigma' := \LexBFS((v_1, \ldots, v_j, \ldots, v_i, \ldots, v_\omega, \ldots), E)\ ,
      $$
      and $D' \in \mathbf{D}(G)$ which is obtained by orienting the edges of $G$ according to $\sigma'$.  We then have $D' \sim_\mcI D$:
      \begin{itemize}[nolistsep]
        \item $D$ and $D'$ obviously have the same skeleton, and both have no v-structures.
        
        \item $D^{(I)}$ and $D'^{(I)}$ have the same skeleton because all arrows between a vertex $a \in I$ and another one $b \notin I$ point away from $a$.
      \end{itemize}
      For $v_i, v_j \in C \setminus I$, the argument is analogous, which proves the claim.
      
      $\mathcal{E_I}(D)$ contains the cliques $C \cap I$ and $C \setminus I$ of size $k$ and $\omega(G) - k$ though.  The fact that $\max\{k, \omega(G) - k\} \geq \lceil \omega(G) / 2 \rceil$ completes the proof.
  \end{subprop}
\end{proof}

\subsection*{References}

\bibliography{dag-active-learning}

\end{document}